\documentclass[letterpaper, 10 pt, conference]{IEEEconf}  

\IEEEoverridecommandlockouts                              
\title{\LARGE \bf
Lettuce modelling for growth control in precision agriculture
}

\author{William Rohde$^{1}$ and Fulvio Forni$^{1}$
\thanks{This work was supported by the Engineering and Physical Sciences Research Council and AgriFoRwArdS CDT [EP/S023917/1], by Cambs Farms Growers Limited (G’s Growers), and by
the project ‘Agri-station for Automation and Growth Optimization’, within the Observatory for Human-Machine Interaction at the University of Cambridge.}
\thanks{$^{1}$Department of Engineering, University of Cambridge  {\tt\small wr281@cam.ac.uk}, {\tt\small f.forni@eng.cam.ac.uk}}%
}

\usepackage{subfig}
\usepackage{graphicx}      
\usepackage{amssymb}
\usepackage{amsmath}
\usepackage{color}

\newtheorem{thm}{Theorem}[section]

\begin{document}

\maketitle
\thispagestyle{empty}
\pagestyle{empty}

\begin{abstract}
  Improving the efficiency of agriculture is a growing priority due to food security issues, environmental concerns, and economics. 
  Precision agriculture and variable rate application technology could enable increases in yield while maintaining or reducing fertiliser use. 
  However, this requires the development of control algorithms which are suitable for the challenges of agriculture. 
  In this paper, we propose a new mechanistic open model of lettuce growth for use in control of precision agriculture.
  We demonstrate that our model is cooperative and fits well to experimental data. 
  We use the model to show, via simulations, that a simple proportional distributed control law increases crop uniformity and yield without increasing nitrogen use, even in the presence of sparse actuation and noisy observations.   
\end{abstract}

\section{Introduction}

The agriculture industry faces several challenges in the near future: global population growth requires increased food production, increased water scarcity requires optimisation of irrigation, and the pollution and emissions associated with fertiliser, pesticides and fossil fuel use need to be minimised.
Precision agriculture offers a potential solution to these challenges, by allowing us to apply actions traditionally performed uniformly across the field at a scale approaching individual management for each plant.
This enables the optimisation of the resource use.

From a systems perspective, precision agriculture can be modelled by a typical feedback control system, 
as shown in Fig.~\ref{fig:PA_feedback}. 
Actuation is provided by variable rate application (VRA) technology, which enables the amount of fertiliser or pesticide (or any other controlled quantity) to be adjusted for different locations in the field.
Sensing the size of individual plants in the field can be achieved through drone imagery, the addition of sensors to traditional farming equipment, and robotic field scouting. 
At present, where precision agriculture is implemented, the decision making process in the precision agriculture feedback loop is rarely automated with crop management decisions being performed by expert growers. 

The design of controllers to perform this automated decision making has seen some interest (see \cite{Cobbenhagen2021-OpportunitiesControlInPrecisionAgriculture}). 
However further research in this area is needed.
A lack of algorithms for automated decision making in VRA of nitrogen fertiliser was listed as a limiting factor by Padilla et al. \cite{padilla2018-proximalOpticalSensorsForNitrogenReview}. 
Similarly, Bhakta et al. \cite{Bhakta2019-StateOfTheArtTechnologyInPrecisionAgriculture} highlight the lack  of a controller that is suitable for the full range of VRA, namely: variable rate applications of fertiliser, pesticide, herbicide, and irrigation.

Applying control to crops in the agriculture sector differs significantly from well-established industrial control scenarios. 
In a typical control process, the feedback loop updates the control signal based on errors between the desired output and observed output at a high frequency,
to reduce the effects of disturbances and modelling uncertainties. 
When applying control in precision agriculture this is less feasible. 
The spatial scale and economic constraints of agriculture prevent machinery from running through the field at a high frequency. 
Currently it is unlikely that the VRA machinery could be run through the field more than a few times during a crop's growth. 
Similarly, sensing at the scale of individual plants is constrained due to the high cost in comparison to the small profit margins on agricultural produce.



\begin{figure}[htbp]
  \begin{center}
  \includegraphics[clip, trim=0.4cm 0.2cm 0.4cm 0.45cm, width=\columnwidth]{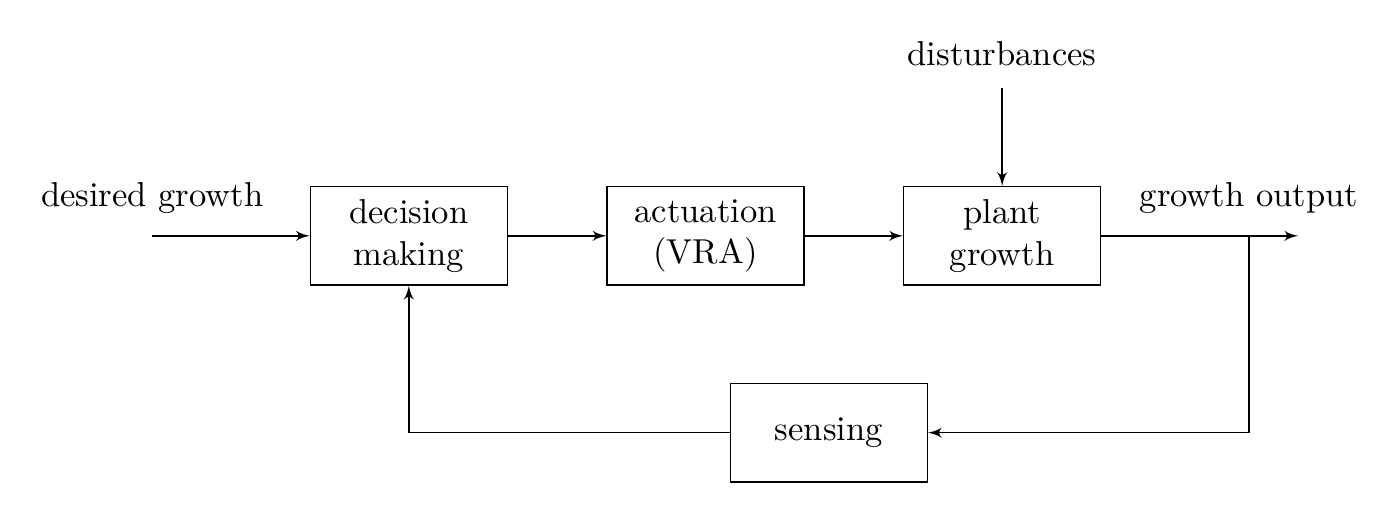}    
  \caption{Precision agriculture as a feedback process} 
  \label{fig:PA_feedback}
  \end{center}
  \end{figure} \vspace{-3mm}
  
The main contribution of this paper is the development of a new model for the growth of lettuce crop, with desirable properties for control at the scale of an agricultural field. Our cooperative compartmental model captures the dynamics of single plants using a reduced number of states. This enables large scale simulations.
Our model is validated using data from controlled lab experiments and field data from farmers. 
Cooperativity guarantees that an increase (decrease) of the nitrogen inputs induces an increase (decrease) of plant biomass outputs, 
a feature that simplifies control design. This is illustrated with a preliminary simulation study based on a distributed control strategy inspired by multi-agent consensus algorithms. 
Results show improvements at harvest even when realistic actuation and sensing constraints 
are considered.

In Section \ref{section:review} we briefly review the current modelling approaches in agriculture. 
Our model is introduced in Section \ref{section:model}. We demonstrate its cooperativity and show that the model is expressive enough to fit experimental data.
Section \ref{section:control} develops simulations for the  control of lettuce plants using proportional distributed feedback with actuation and sensing constraints.
Conclusions and future research directions are summarised in Section \ref{section:conclusion}.

\section{A review of models in agriculture} \label{section:review}

In agriculture, plants are managed by expert growers as open systems. 
Through interventions in the field, the plant's ``inputs", such as water or fertiliser, 
are driven to cause a change in the ``output'' yield at harvest, with weather conditions acting as disturbances. 
The growers' actions are usually driven by robust average ``personal" models based on direct practice and experience.
These contrast with the modelling practice of plant science, which focuses on the fine details of plant metabolism and structure. 
This difference is justified by the fact that 
agriculture is characterised by data scarcity and uncertainties due to genetic and environmental differences at the individual plant level.
In this environment, the predictive ability of high resolution scientific models is often 
compromised by the challenge of noisy measurements and uncertainties.

This divide is not new to feedback control theory. The first step in control is often to build a model that captures 
the fundamental features of a process within a suitable compromise between accuracy and complexity.
In this section we review existing mechanistic and functional-structural approaches to plant modelling. 
These models build a ``virtual'' plant to probe scientific hypothesis and to enable experimentation, with the
goal of reaching a deeper analytical understanding of the plant behaviour, \cite{HILTY2021-PlantGrowthWhatHowWhy}.  
They provide the basis for the derivation of a simplified lettuce model for agricultural growth control in Section \ref{section:model}. 

\textbf{Mechanistic plant models} mathematically describe the processes within the plant and their interaction with the environment. These models
are typically used to estimate yields (see \cite{PELAK2017DynamicalSystemsFrameworkCropModels}, \cite{ Pearson1997_modelEffectsEnvironmentOnGrowthLettuce}), predict harvest times, and predict crop growth. 
They can be designed to respond to environmental conditions such as weather, temperature, light conditions, fertiliser treatments, and irrigation. 
Mechanistic models can scale from individual plants to ecosystems by assuming linear superposition and by aggregating model states.

At the individual plant scale mechanistic models such as those produced by Thornley \cite{THORNLEY1998-ModellingShootRootRelations}, Pearson et al.  \cite{Pearson1997_modelEffectsEnvironmentOnGrowthLettuce}, and Harwood et al. \cite{HARWOOD2010_ModellingUncertaintyIceberg}, are often compartmental models which describe the flow of resources such as carbon and nitrogen through the plant. 
Pearson and Harwood model the flow of carbon through a plant using two compartments: one for carbon storage, and the second representing structural carbon, with photosynthesis modelled as a carbon inflow to the storage compartment. 
Their models also include the effect of temperature on the growth of structural biomass, photosynthesis, and senescence. 

The Thornley model includes both nitrogen and carbon flows. 
The plant is modelled with six compartments divided between the shoots and roots of the plant.
The shoots and roots each have a nitrogen store, carbon store and structural biomass. 
The plant absorbs nitrogen into the ``root nitrogen store" which flows with some resistance into the ``shoot nitrogen store". 
Similarly, the plant has a carbon inflow to the ``shoot carbon store" which flows to the ``root carbon store".
Carbon and nitrogen stores contribute to the development of structural biomass, with the root stores contributing to the growth of structural biomass in the roots, and shoot stores contributing to growth in the shoots.

\textbf{Functional-structural models} combine mechanistic elements with the morphology of plants. 
They offer a higher resolution simulation of individual plants compared to purely mechanistic models.
These models can be used to characterise the growth of individual organs of the plant, or specific behaviours such as self-competition, or detailed physical phenomena such as transport of water and carbon \cite{Lacointe2019-PiafMunch,Zhou2020-CPlantBox}.
As a result, functional-structural models are used as digital twins for simulation.
The resolution of these models can provide insights that would be challenging to observe in real plants, but this comes at the cost of increased computational complexity. 

While many functional-structural models include the full 3D geometry of the plant, some use a simpler topological representation. 
Topological models use graph-like structures with nodes and edges to represent plants.
Topological models are less computationally complex than those based on the 3D geometry of the plant.

A functional-structural model may be necessary to capture the input/output dynamics of some plants.
This is particularly important for plants where fruits are the output of interest.
In these cases, characterisation of the relevant output requires knowledge of plant morphology. 
However, the use of functional-structural models may be intractable for the real-time control of crops, where hundreds of thousands of plants may be associated to a model with a large state-space. 
Mechanistic compartmental models for individual plants are much less computationally complex than functional-structural models, while still capturing key information about the system. 
As mechanistic models describe the processes within plants they can be easily extended to open models, with controlled inputs, exogenous disturbances (environment), and measured outputs.

\section{Modelling for control} \label{section:model}
\subsection{Model description}

We propose an open compartmental mechanistic plant model designed for feedback control of
crop growth using variable-rate nitrogen application.
Our model is a simplification of the root-shoot growth model by Thornley \cite{THORNLEY1998-ModellingShootRootRelations}. 
Several functions are adapted to make the model open, capturing the effect of temperature and light variations from the environment and of nitrogen variations in the soil.

In the Thornley model there are six compartments divided between the plant's shoots (leaves and stems) and roots, so that each has a carbon store, nitrogen store, and structural biomass.
In practice, the measurement of the root compartments in agriculture is difficult. 
These compartments have to be estimated indirectly from noisy observations of the shoots, in the presence of model uncertainties. 
This undermines the level of detail the Thornley model provides.
Our model merges the roots and shoots resulting in a three compartment model, and uses a parameter to partition the biomass into shoots and roots, rather than estimating root states.  

In our model, the structural biomass of the plant is a single compartment, $\boldsymbol{b} \geq 0$. 
The portion of the plant's structural biomass that is attributed to the shoots is given by the parameter $\psi \in (0,1)$, so that $\psi \boldsymbol{b}$ is the shoot dry biomass of the plant and $(1-\psi)\boldsymbol{b}$ is the root dry biomass of the plant.
The plant has a carbon storage compartment, $\boldsymbol{c} \geq 0$, and a nitrogen storage compartment, $\boldsymbol{n} \geq 0$.  
This gives the (non-negative) state vector $\boldsymbol{x}^T = \begin{bmatrix} \boldsymbol{b},& \boldsymbol{c},& \boldsymbol{n} \end{bmatrix}$. 
The uptake of nitrogen from the soil creates an inflow to the $\boldsymbol{n}$ compartment, and photosynthesis provides an inflow to the $\boldsymbol{c}$ compartment. 
Both nitrogen and carbon are consumed in the production of structural biomass $\boldsymbol{b}$.
The flows between the compartments in the model are illustrated in Fig. \ref{fig:custom_model_compartments}.

\begin{figure}[h]
  \begin{center}
  \includegraphics[clip, trim=0.28cm 0.2cm 0.2cm 0.3cm,width=1\columnwidth,page=8]{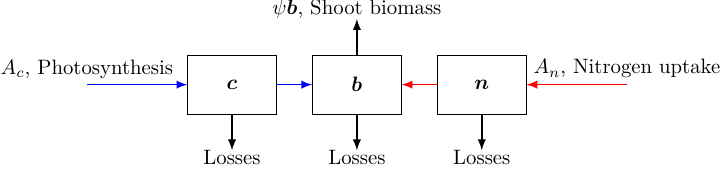}    
  \caption{The compartmental lettuce model: The red arrows show the flow of nitrogen through the system and the blue arrows show the flow of carbon.} 
  \label{fig:custom_model_compartments}
  \end{center}
  \end{figure}

\noindent  
The dynamics of an individual plant are described by 
\begin{subequations} \label{eqn:model}
  \begin{align} 
    \dot{\boldsymbol{b}} &= G(\boldsymbol{b},\boldsymbol{c}, \boldsymbol{n},\boldsymbol{T}) - L_b(\boldsymbol{b})\\
    \dot{\boldsymbol{c}} &= A_c(\boldsymbol{b},\boldsymbol{c},\boldsymbol{I}) - B_c(\boldsymbol{c},\boldsymbol{T}) \label{eqn:dc}\\
    \dot{\boldsymbol{n}} &= A_n(\boldsymbol{b},\boldsymbol{n},\boldsymbol{u}) - B_n(\boldsymbol{n},\boldsymbol{T}) \,, \label{eqn:dn}
  \end{align}
\end{subequations}
where 
\begin{subequations} \label{eqn:system-equations} 
  \begin{align} 
    G(\boldsymbol{b},\boldsymbol{c}, \boldsymbol{n}, \boldsymbol{T}) &= k R_T(\boldsymbol{T})\frac{\boldsymbol{c}}{\boldsymbol{b}}\frac{\boldsymbol{n}}{\boldsymbol{b}}\boldsymbol{b} \label{eqn:G}\\
    R_T(\boldsymbol{T}) &= \dfrac{T_\text{op} - |T_\text{op} - \boldsymbol{T}|}{T_\text{op}} \label{eqn:Te}\\
    L_b(\boldsymbol{b}) &= \dfrac{k_{l}\boldsymbol{b}}{1 + \frac{k_{ml}}{\boldsymbol{b}}}\label{eqn:Lb}\\
    B_c(\boldsymbol{c},\boldsymbol{T})&=\theta_c k R_T(\boldsymbol{T}) \frac{\boldsymbol{c}}{\boldsymbol{b}}\boldsymbol{b} \label{eqn:bc}\\
    B_n(\boldsymbol{n},\boldsymbol{T})&=\theta_n k R_T(\boldsymbol{T})\frac{\boldsymbol{n}}{\boldsymbol{b}}\boldsymbol{b} \label{eqn:bn} \\
    A_{c}(\boldsymbol{b,c,I})&=\dfrac{\sigma_c \psi\boldsymbol{b}  \boldsymbol{I}}{(1+\frac{\boldsymbol{\psi b}}{v})(1+\frac{\boldsymbol{c}}{\boldsymbol{\psi b} j_c})}\label{eqn:Ac}\\
    A_{n}(\boldsymbol{b,n,u})&=\dfrac{\sigma_n (1-\psi)\boldsymbol{b}  \boldsymbol{u}}{(1+\frac{\boldsymbol{(1-\psi) b}}{v})(1+\frac{\boldsymbol{n}}{\boldsymbol{(1-\psi) b} j_n})} \ .\label{eqn:An}
  \end{align}
\end{subequations}
Parameters are described in Table \ref{tab:params}.

The growth of structural biomass is given by the function $G(\boldsymbol{b},\boldsymbol{c}, \boldsymbol{n},\boldsymbol{T})$ in \eqref{eqn:G},  where $\boldsymbol{T}$ is temperature. 
The growth rate is proportional to the mass concentration of carbon in the plant, $\boldsymbol{c}/\boldsymbol{b}$, the mass concentration of nitrogen, $\boldsymbol{n}/\boldsymbol{b}$, and the plants current biomass, $\boldsymbol{b}$. 
We model temperature dependence in the structural growth of the plant by multiplying the rate with a unit triangle function, $R_T(\boldsymbol{T})$ in \eqref{eqn:Te}, centred on an optimal temperature, $T_\text{op}$. 
This models the effect of the growth rate peaking at an optimal temperature and declining away from it.
This is similar to the temperature dependence used by Pearson et al. in their lettuce growth model \cite{Pearson1997_modelEffectsEnvironmentOnGrowthLettuce}. 
The biomass loss equation, $L_b(\boldsymbol{b})$ in \eqref{eqn:Lb}, models the biomass loss rate as linear when $\boldsymbol{b}$ is large in comparison to $k_{ml}$, but with a reduction in loss when the $\boldsymbol{b}$ is small. 

The contributions from the carbon and nitrogen compartments to growth of structural biomass are modelled by the terms $B_c(\boldsymbol{c},\boldsymbol{T})$ and $B_n(\boldsymbol{n},\boldsymbol{T})$ in \eqref{eqn:model}.
These terms represent the loss of carbon and nitrogen due to the growth of the structural biomass and to other complex phenomena such as plant respiration and natural degradation.
The parameters $\theta_c$ and $\theta_n$ are the rate at which carbon and nitrogen respectively are consumed during these processes.
As in \eqref{eqn:G}, these functions are proportional to the mass concentration of the substrate and to the biomass of the plant.
To illustrate this the redundant $\boldsymbol{b}/\boldsymbol{b}$ is included in \eqref{eqn:bc} and \eqref{eqn:bn}.

The nitrogen uptake equation, $A_n(\boldsymbol{b},\boldsymbol{n},\boldsymbol{u})$ in \eqref{eqn:An},  and the photosynthesis equation, $A_c(\boldsymbol{b},\boldsymbol{c},\boldsymbol{I})$ in \eqref{eqn:Ac}, are similar. 
These functions have been adapted from the Thornley model, by adding proportional input variables to model the response to environmental conditions.
The carbon inflow due to photosynthesis is proportional to the light input to the plant, $\boldsymbol{I}$, and the shoot structural biomass $\psi \boldsymbol{b}$.
Similarly, the nitrogen inflow is proportional to the nitrogen available to the plant, $\boldsymbol{u}$, and the root structural biomass $(1-\psi)\boldsymbol{b}$.
The first term in the denominator, $(1+\frac{\boldsymbol{\psi b}}{v})$ in \eqref{eqn:Ac}, represents the saturation of nutrient assimilation due to the plant's size and models effects such as self-shading.
The second term, $(1+\frac{\boldsymbol{c}}{\psi \boldsymbol{b} j_n})$, models a saturation in nutrient assimilation due to a high mass concentration of the nutrient. 

\begin{table}[htbp]
  \captionsetup{width=1\columnwidth}
 \caption{Model parameters and their descriptions.
  Value examples are from the ``Good fit" timeseries in Figure \ref{fig:paramfit}.}  \label{tab:params}
 \begin{center}
   \begin{tabular}{cclc} 
     \hline
     Parameter  & Description & Value &Units\\
     \hline
     $k$ & Structural plant growth  & 1000 & $\text{s}^{-1}$\\
     $k_l$ & Litter loss rate & 0.149  & $\text{s}^{-1}$\\
     $k_{ml}$ & Litter saturation parameter  & 0.0221 & $\text{g}$\\
     $\sigma_c$ & Carbon assimilation rate & 0.260 & $\text{m}^{2}\text{W}^{-1} \text{s}^{-1}$\\
     $\sigma_n$ & Nitrogen assimilation rate & 70.0 & $\text{g}^{-1}\text{s}^{-1}$ \\
     $v$ & Self-shading saturation value  & 0.0620 & $\text{g}$\\
     $j_c$ & C product inhibition & 0.144 & -\\
     $j_n$ & N product inhibition  & 0.115 & -\\
     $\psi$ & Shoot portion of biomass & 0.718 & - \\
     $T_\text{op}$  & Optimal temperature & 22 & $^\circ\text{C}$\\
     $\theta_c$  & $\boldsymbol{c}$ biomass contribution rate & 6.89e-2& -\\
     $\theta_n$  & $\boldsymbol{n}$ biomass contribution rate & 5.57e-6& -\\
     \hline
   \end{tabular}
   \end{center}
 \end{table}

The control input, $\boldsymbol{u} \geq 0$, is the nitrogen availability in the soil.
The output of the model is the portion of the plant's biomass that belongs to the shoots of the plant, $\boldsymbol{y}=\psi \boldsymbol{b} \geq 0$.    
Temperature, $\boldsymbol{T}$, and light, $\boldsymbol{I}\geq 0$, are disturbance inputs to the model. 
For simplicity, in what follows we model them as time-varying parameters.
The system of differential equations can be represented in compact form as the single input single output system \eqref{eqn:statespace_combinedmodel} below, 
where  $C = \begin{bmatrix}\psi & 0 & 0 \end{bmatrix}$
and $f$ collects the right-hand side of \eqref{eqn:model}, with the parameter vector represented by $\theta$.
\begin{subequations} \label{eqn:statespace_combinedmodel}
  \begin{align}
    \dot{\boldsymbol{x}} &= f(\boldsymbol{x},\boldsymbol{u}; \theta) \\
    \boldsymbol{y} &= C\boldsymbol{x}
  \end{align}
\end{subequations}
\subsection{Monotonicity} \label{section:monotonicity}

If we restrict the input space of the model to the controlled nitrogen, 
by considering light and temperature as simple time-varying parameters, 
we can show that our model is an open cooperative system \cite{Angeli2003,smith1995-monotone}. 
This has several beneficial properties for control, the most significant of which 
is that an increase of the nitrogen input $\mathbf{u}$ leads to an increase of the shoot biomass $\mathbf{y}$.
This is important for control as the monotone relationship is preserved despite large uncertainties on the system parameters.

\begin{thm}
    \label{thm:model_cooperativity}
    The model \eqref{eqn:model} and \eqref{eqn:system-equations} described by the system \eqref{eqn:statespace_combinedmodel} 
    is an open cooperative system. 
\end{thm}

\begin{proof} 
    By the Kamke conditions \cite{smith1995-monotone,Angeli2003}: the system is cooperative if $C \ge 0$ (element-wise),
    for all $i$
    $$
    e_i^T \frac{\partial f(\mathbf{x},\mathbf{u}; \theta)}{\partial \mathbf{u}} \ge 0 \,,
    $$ 
    and for all $i \neq j$
    $$
    e_i^T \frac{\partial f(\mathbf{x},\mathbf{u}; \theta)}{\partial \mathbf{x}} e_j \ge 0 \,,
    $$
    where $e_i$ is the $i$-th column of the  identity matrix.

    By inspection of the Jacobian matrices for our model, every off-diagonal element is positive, therefore the system is cooperative.
    This is illustrated in Figure \ref{fig:cooperativity-graph}, where the off-diagonal elements of the Jacobian matrices are plotted as the edges of a directed graph with each edge marked with its sign.
    The positivity of the edge $(\mathbf{u},\mathbf{n})$ is given by
    $\frac{\partial A_n}{\partial u} \geq 0$. 
    The positivity of the edge $(\mathbf{n},\mathbf{b})$ is given by 
    $\frac{\partial G}{\partial n} \geq 0$ and $\frac{\partial A_n}{\partial b} \geq 0$. 
    The latter follows from the fact that 
    $$\frac{\partial A_n}{\partial \boldsymbol{b}} = \frac{\boldsymbol{u} \boldsymbol{b}  v j_n \sigma_n \gamma^2 (2 \boldsymbol{n} v + \boldsymbol{b} \boldsymbol{n} \gamma + \boldsymbol{b} v j_n \gamma) }{(\boldsymbol{n} + \boldsymbol{b} j_n \gamma)^2 (v + \boldsymbol{b} \gamma)^2}\ge 0,$$
    where $\gamma = (1-\psi) > 0$.
    Likewise, the positivity of the edge $(\mathbf{c},\mathbf{b})$ is given by 
    $\frac{\partial G}{\partial c} \geq 0$ and $\frac{\partial A_c}{\partial b} \geq 0$. 
    The latter follows from the fact that 
    $$\frac{\partial A_c}{\partial \boldsymbol{b} }  =  \frac{  \boldsymbol{I}  \boldsymbol{b}  v j_c  \sigma_c \psi^2(2 \boldsymbol{c}  v + \boldsymbol{b}  \boldsymbol{c}  \psi + \boldsymbol{b}  v j_c \psi)}{(\boldsymbol{c} + \boldsymbol{b}  j_c \psi)^2 (v + \boldsymbol{b}  \psi)^2}  \ge 0.$$
    Finally, the positivity of the edge $(\mathbf{b},\mathbf{y})$ follows from $C \geq 0$. 
\end{proof}

\begin{figure}[h]
  \begin{center}
  \includegraphics[clip, trim=0cm 1cm 0cm 1.1cm ,width=0.95\columnwidth]{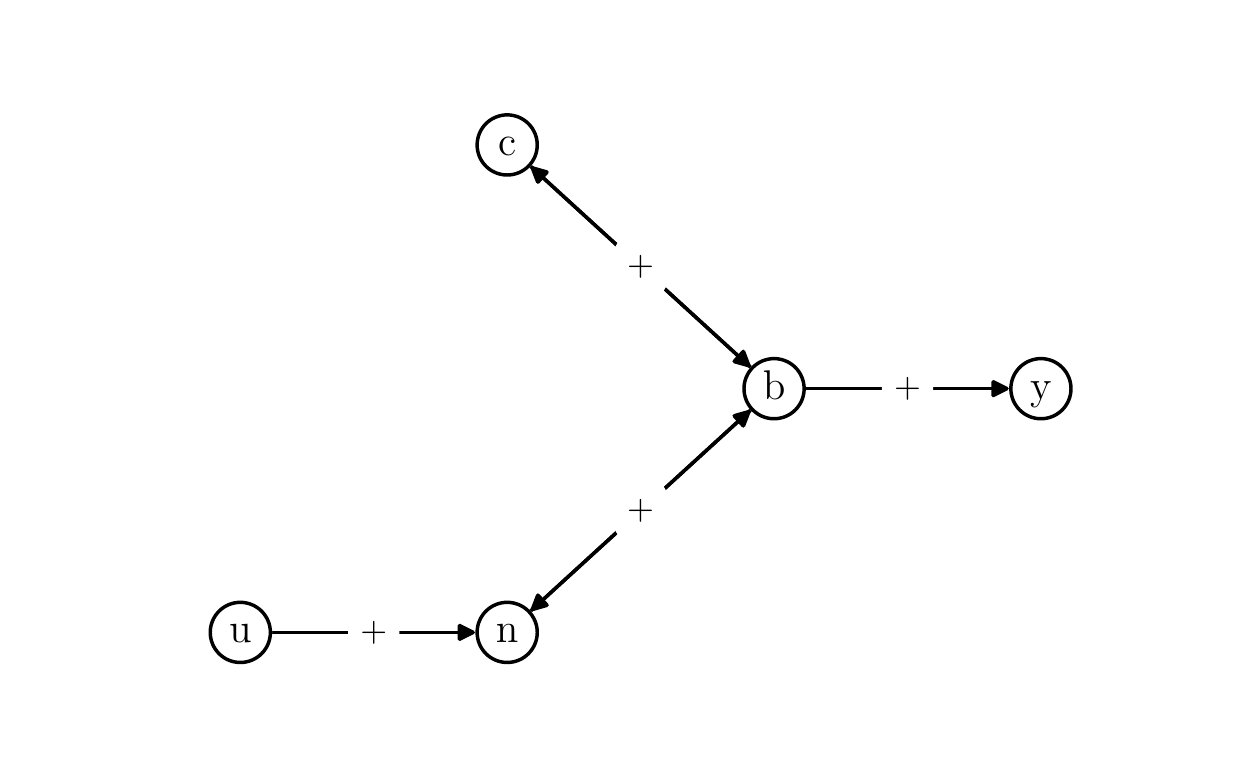}    
  \caption{Sign graph associated to the linearization of \eqref{eqn:statespace_combinedmodel}.} 
  \label{fig:cooperativity-graph}
  \end{center}
  \end{figure}

\subsection{Parameter Fitting}

In simulation, the evolution in time of the biomass follows a typical logistic function as shown in Fig. \ref{fig:paramfit}.
This logistic curve captures the initial slow start of a plant, followed by a period of intense growth activity, which eventually slows down before harvest. 
We tested the ability of our model to capture realistic growth by fitting the model to experimental data from lettuce plants: first, a dataset of plants grown in an indoor grow tent; second, a dataset of plants from a real agricultural environment.

We grew 19 iceberg lettuce plants in an indoor grow tent in the Department of Engineering at the University of Cambridge, and produced a time series for the area of each plant based on segmented images.
We converted this to dry biomass using the area and mass at harvest, a factor of 0.1 was used to convert from fresh biomass to dry.
We fit our model to this dataset with the results shown in Figure \ref{fig:paramfit_fullFigure}. 
The parameter fitting for the models was performed using the python function \emph{scipy.optimize.minimize}, which implements the Broyden–Fletcher–Goldfarb–Shann algorithm with box constraints, \cite{byrd1995-LBFGSB}, with a mean least squares cost function. 
We calculated the residuals between the model predicted shoot biomass $\boldsymbol{y}=\psi \boldsymbol{b}$ and the observed shoot biomass of the dry biomass timeseries. 
The normalised root mean square error (NRMSE) for all 19 lettuce plants are summarised in the histogram in Figure \ref{fig:paramfitagripod_hist}. 
Figure \ref{fig:paramfit} shows a plot of the trajectory and observed data for three example cases which represent good, average, and poor fitting.
An example of fitted parameter values are given in Table \ref{tab:params}.
Despite the limitations of our small dataset, Figure \ref{fig:paramfitagripod_hist} illustrates that our model has enough flexibility to capture the growth of most plants, with few outliers.

To assess the ability of the model to fit to realistic agricultural data, which has a low sample rate, we fitted the model to a dataset of 1000 batches of outdoor lettuce plants. 
This dataset was provided by G's Growers (a major growers' cooperative and lettuce producer in the United Kingdom).
Each entry in this dataset consists of sampled plants that were picked and weighed throughout the growing period of the batch. 
Each timeseries in the dataset had between 3 and 12 data points (a good illustration of the sampling frequency expected in outdoor agriculture).
The wet mass of the plants was converted to a dry biomass using a factor of 0.1 and used to fit the model.
A histogram of the normalised root mean square error (NRMSE) between the model and data is shown in Figure \ref{fig:paramfithist}. 
A selection of raw data and fit responses are shown in Figure \ref{fig:paramfithist_examples}.
Figure \ref{fig:paramfithist_hist} shows that, in the majority of cases, our model performs well with a low NRMSE and the ``average fit'' in the dataset fits the available data well.

These results show that our model is capable of capturing the behaviour of annual plants in which the growth of the plant is the output of interest. 
The parameter fitting has been performed on lettuce, but we believe that similar results can be obtained on other crops where the size of individual plants is directly related to yield, such as salads and brassicas.

\begin{figure}[h]
  \begin{center}
    \vspace{2mm}
  \subfloat[][Model fit to lettuce data. Data from a lettuce plant grown in an indoor grow tent. Cross marks are datapoints.
  Solid curves represent simulated shoot biomass $\boldsymbol{y}=\psi \boldsymbol{b}$ after fitting.
  \label{fig:paramfit}]{\includegraphics[width=0.83\columnwidth]{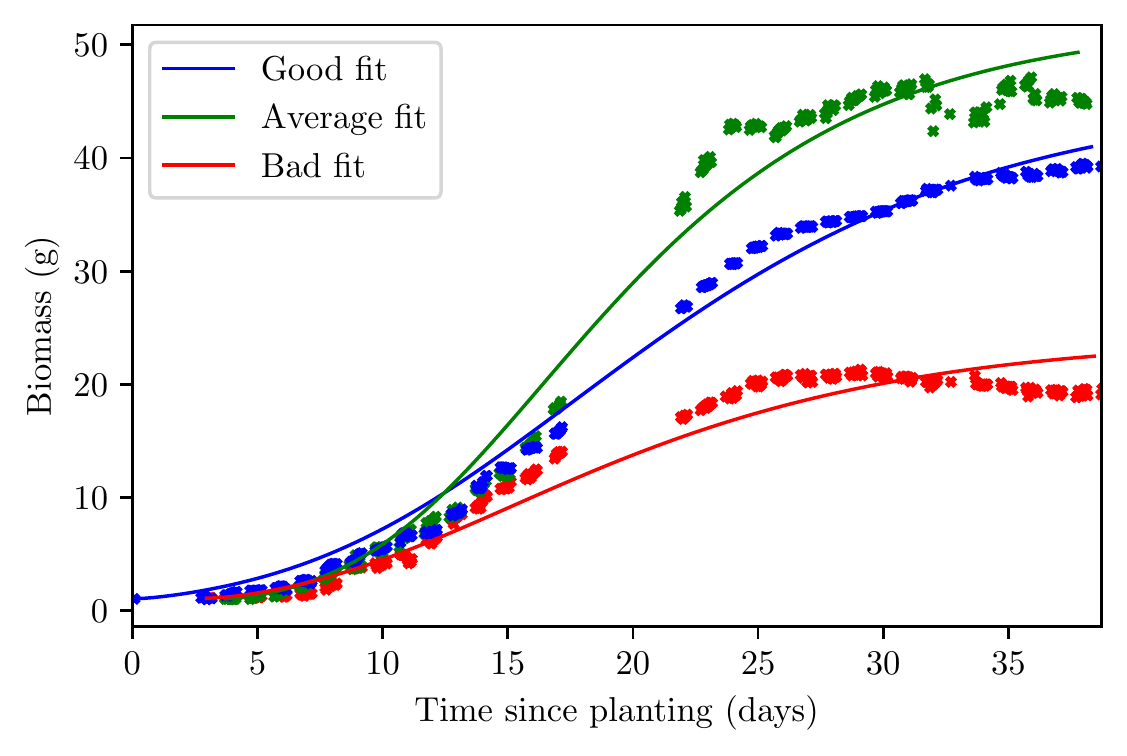} }   
  \\
  \subfloat[][A histogram of the NRMSE between fit models and observed data for the dataset collected for plants grown in an indoor grow tent. The coloured lines correspond to the examples in Figure \ref{fig:paramfit}. \label{fig:paramfitagripod_hist}]{\includegraphics[width=0.83\columnwidth]{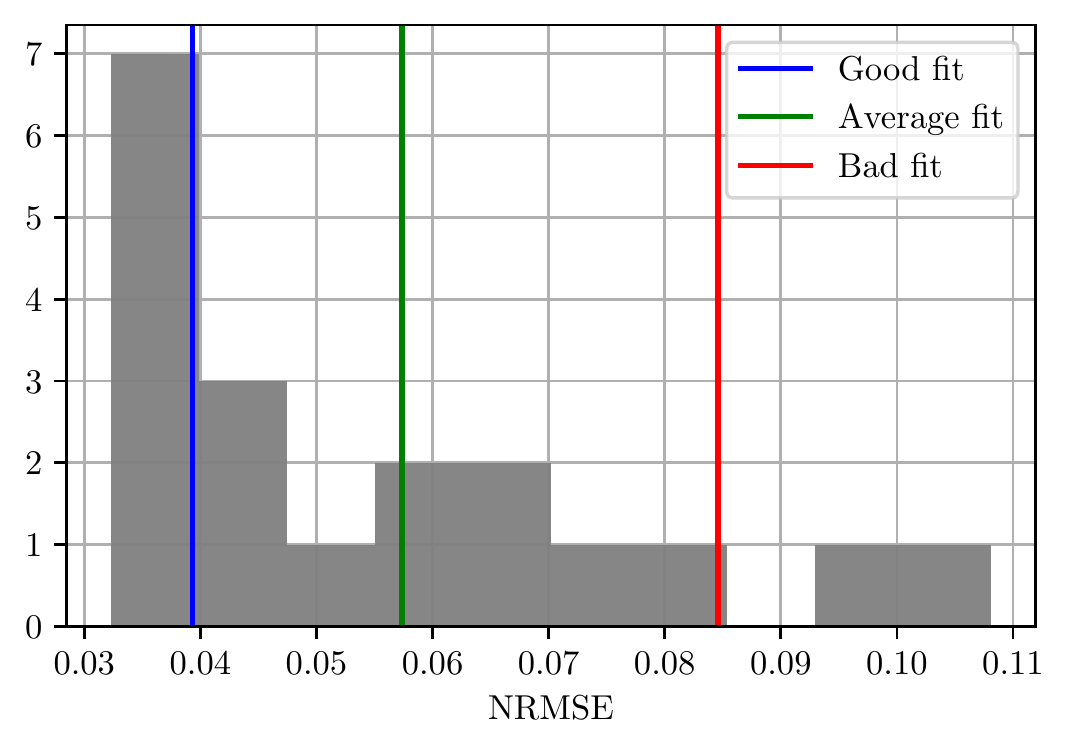}   }   
  \caption{The evolution of the model fit to data collected in an indoor grow tent  and the NRMSE for the full dataset.} 
  \label{fig:paramfit_fullFigure}
  \end{center}
  \end{figure}

\begin{figure}[h]
  \begin{center}
    \vspace{2mm}
  \subfloat[][Three examples of parameter fitting for the agricultural dataset. The solid line shows the trajectory of the simulated model output $\boldsymbol{y}=\psi \boldsymbol{b}$ while cross marks are datapoints.  \label{fig:paramfithist_examples}]{\includegraphics[width=0.84\columnwidth]{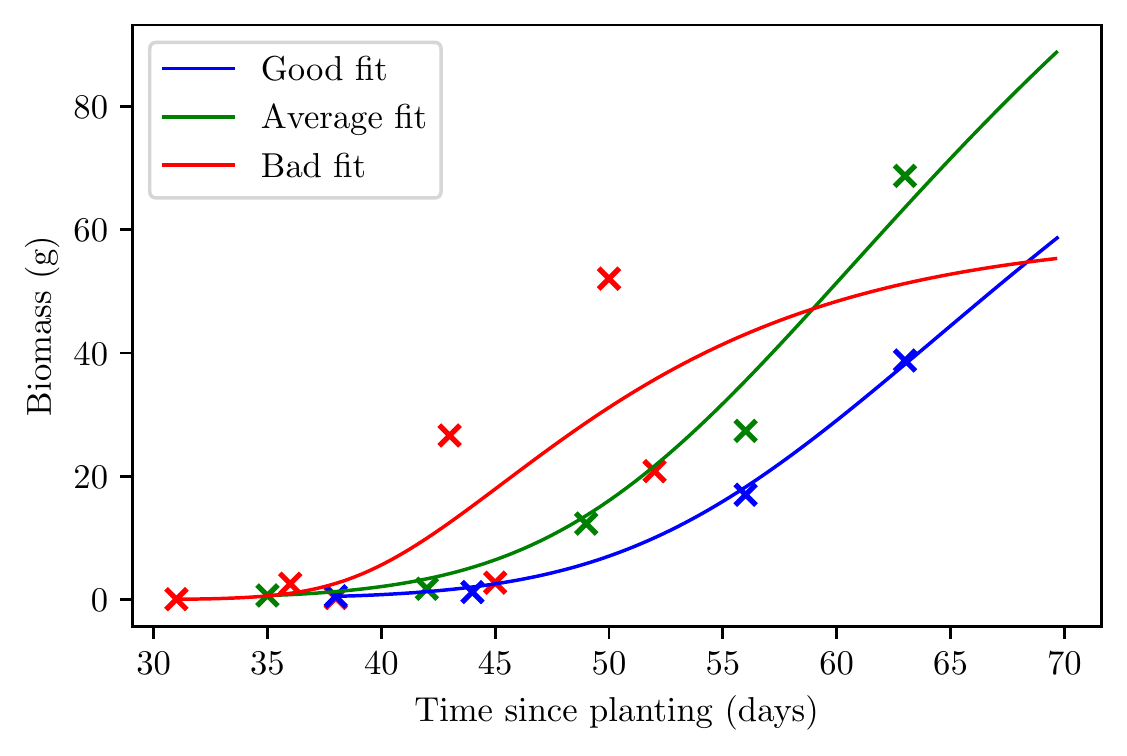}}   
  \\
  \subfloat[][A histogram of the NRMSE between fit models and observed data from an agricultural environment. The coloured lines correspond to the examples in Figure \ref{fig:paramfithist_examples}. \label{fig:paramfithist_hist}]{\includegraphics[width=0.84\columnwidth]{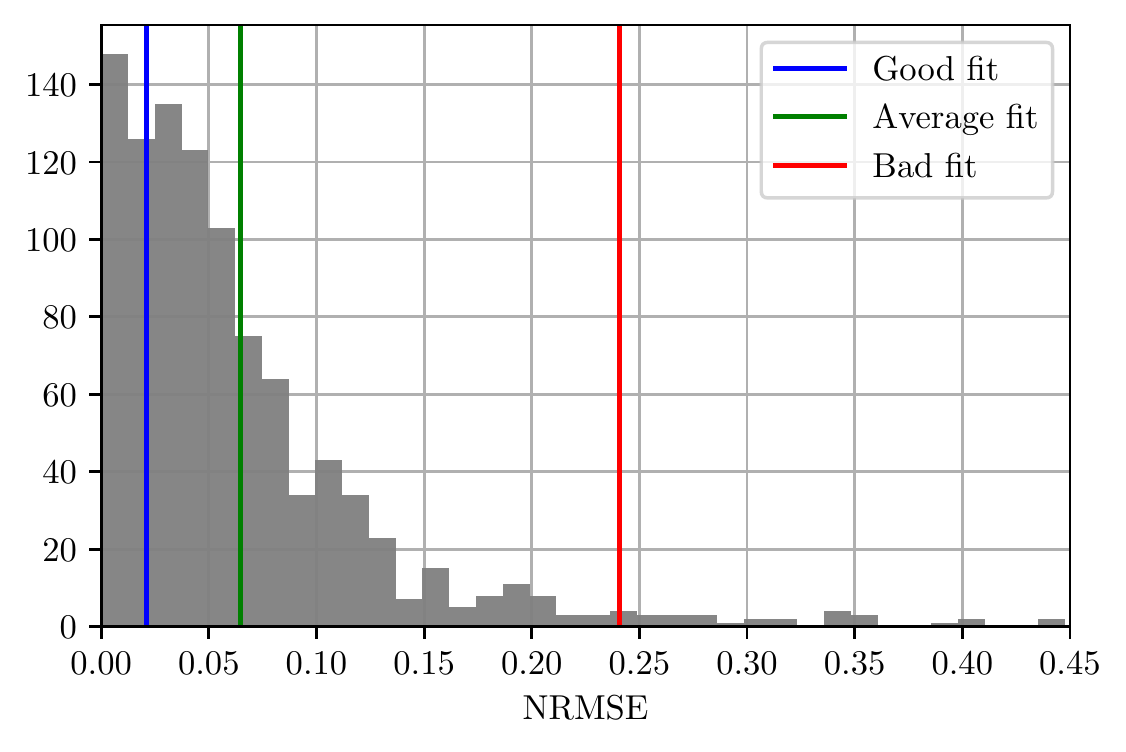}   }   
  \caption{Parameter fitting performance with agricultural data. The examples of ``good" and ``average" fit shows good proximity to the data. The ``bad'' is also
  justified by the large noise affecting the data.} 
  \label{fig:paramfithist}
  \end{center}
  \end{figure}

\section{Control in the field} \label{section:control}

\subsection{The control problem}
The main objective of a farmer is to maximise the profit from harvest.
This can be achieved with an increase in yield and/or with a reduction in resource use. 
For high value produce, such as lettuce or other produce with higher value than grain, a significant portion of the harvested crop is discarded due to failing to meet buyer specifications.
Improving the crop uniformity at harvest thus results in a higher yield, with larger portion of produce meeting the specification.
This is the motivation for our preliminary study below, where we use our model to show how simple variable-rate nitrogen control strategies lead to higher uniformity and can increase yield. 
The study also shows that the improvement in uniformity enables us to reduce the  total amount of nitrogen in comparison to the current agricultural practice based on uniform rate application. 
This is important as it allows a reduction in fertiliser costs and avoids unnecessary pollution.
In summary, the control problem is to \emph{find a control strategy for variable-rate nitrogen application to increase uniformity at harvest and reduce total nitrogen use.}

Accurate regulation of nitrogen can be achieved through variable rate application.
Current approaches are either heuristic or based on optimisation. 
Heuristic approaches divide the field into management zones which are given different treatments based on either the size of plants within them or historic data.
Guerrero and Mouazen carried out a simulated experiment in which they assigned fertility classes to management zones based on satellite imagery, historic yield data, and soil data. 
Additional fertiliser was then applied to lower fertility classes and resulted in 33-56.2\% reduction in fertiliser use in comparison to the uniform scenario \cite{GUERRERO2021_EvaluationVRANscenariosInCeralCrops}.

Optimisation based methods have also been shown to be effective and are not necessarily confined to nitrogen.
Wu et al.  performed optimal control for the water supply to virtual sunflower plants \cite{Wu2012-OptimalControlPlantgrowthWaterSupply}.
Neto et al.  computed the optimal nitrogen supply for growth of crop and designed a controller for an industrial sludge plant to implement this \cite{Neto2021-MPCSludgePlantGrowth}. 
We refer the reader to \cite{Cobbenhagen2021-OpportunitiesControlInPrecisionAgriculture,Schoonen2019-OptimalIrrigationMPC} for further examples.

However, optimal control methods rely on accurate crop models. 
This is a demanding requirement, as large model uncertainties are unavoidable in current agriculture practice.
This is illustrated by Figure \ref{fig:free_plants} which shows the uncontrolled growth of 100 simulated plants based on \eqref{eqn:model}, \eqref{eqn:system-equations}, with a constant level of nitrogen $\bar{u}=0.075$g available to each plant. 
In this plot, model parameters are normally distributed with a standard deviation of 5\% from the nominal values in Table \ref{tab:params}. 
Figure \ref{fig:free_plants} shows how, over the course of the plant's life, these small uncertainties result in significant changes in the growth of the plant.
The dashed black line represents a threshold below which lettuce plants are rejected due to buyer specifications. 
In this paper the threshold was arbitrarily set to the 10th percentile of the final biomass for the uncontrolled scenario. 
This corresponds to  a dry shoot biomass of 36.6g.

\begin{figure}[h]
  \begin{center}
  \includegraphics[width=0.85\columnwidth]{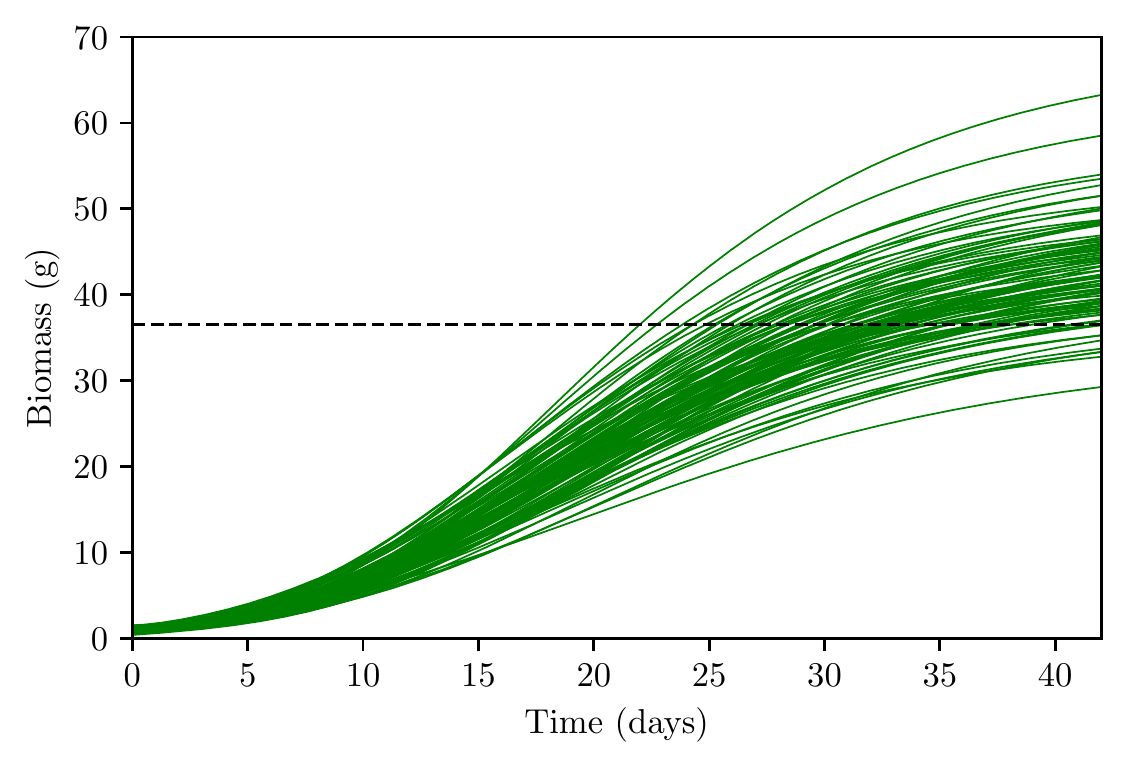}    
  \caption{Uncontrolled growth of 100 simulated plants  \eqref{eqn:model},\eqref{eqn:system-equations}, with $\boldsymbol{u}=\bar{u}=0.075$g for each plant. The black dashed line represents the rejection threshold (36.6g).} 
  \label{fig:free_plants}
  \end{center}
  \end{figure}

\subsection{Consensus-based multi-agent solution}

We propose, an alternative approach to heuristic or optimisation based methods, by considering the field as a multi-agent system on which we deploy consensus algorithms \cite{Olfati-Saber2003-consensusProtocolsforNetworksDynamicAgents,OlfatiSaber2007-ConsensusAndCooperationInNetworkedMultiAgentSystems}. 
The goal is to derive a coordinated variable-rate nitrogen application strategy for each plant to increase uniformity.
We look at plants as agents, with the nitrogen inputs as controllable and affecting their growth, and estimate the dry shoot biomass of the plant through remote sensing.
Even though consensus cannot be achieved through differential nitrogen application alone in a heterogeneous field, the control action will reduce the variance in plant growth and increase uniformity.

The control strategy is rooted on Theorem \ref{thm:model_cooperativity}, which guarantees that an increase in the available nitrogen input will result in an increase in the plants growth and a decrease in the nitrogen input will result in reduced growth, regardless of uncertainties in the model parameters.
This is illustrated in Figure \ref{fig:effec_of_nitrogen}.
These simulation results show the monotone relationship between the plant biomass at day 50 and different constant levels of nitrogen $\mathbf{u}$.
Each curve represents a different set of parameters, highlighting that modelling uncertainties do not affect the order preserving properties of our cooperative model. 

\begin{figure}[h]
  \begin{center}
    \vspace{2mm}
  \includegraphics[width=0.85\columnwidth]{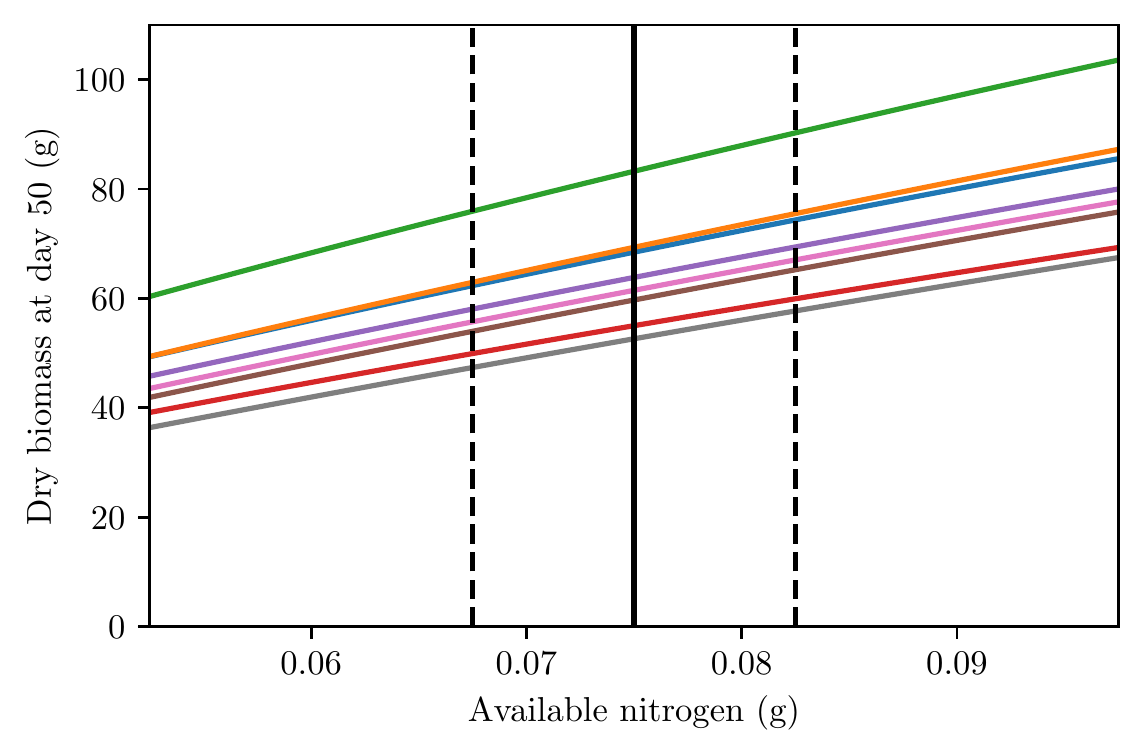}    
  \caption{Sensitivity of biomass to available nitrogen. Effect of constant nitrogen levels on the size of the plant biomass, $\boldsymbol{b}$, at day 50. 
  Each curve is related to a different set of parameters normally distributed about those given in Table \ref{tab:params}.
  The dashed lines represent the saturation function adopted for feedback.
  The solid black vertical line is $\bar{u}$.} 
  \label{fig:effec_of_nitrogen}
  \end{center}
  \end{figure}

The simplest approach is a piecewise constant application strategy for each plant based on proportional feedback. For each plant $i$ we consider the strategy 
\begin{equation} \label{eqn:prop_global}
  \mathbf{u}_i = \bar{u} + \text{sat}(k[\bar{y}-\mathbf{y}_i]).
\end{equation}
where $\bar{u}=0.075$g is a uniform constant value of nitrogen applied to the field, perturbed by 
$\mathrm{sat}(k[\bar{y}-\mathbf{y}_i])$ which penalises the displacement between the current
plant growth $\mathbf{y}_i$ and the field average growth 
$\bar{y} = \frac{1}{N}\sum_{i=1}^N \mathbf{y}_i $. 
The function $\mathrm{sat}(\cdot)$ corresponds to a piecewise linear, symmetric saturation function limiting $u_i$ to the range $0.0675 = \bar{u}-u_\text{range}<u_i<\bar{u}+u_\text{range} = 0.0825$. 

Equation \eqref{eqn:prop_global} applies higher nitrogen levels to plants that are underdeveloped with respect to the average of the field, and reduces nitrogen levels applied to overdeveloped plants.  
Each plant received a nitrogen application of $\mathbf{u}_i$ every day, computed with $k = 0.05$, and modelled as a piecewise constant signal.
Results are illustrated in Figure \ref{fig:control_prop_to_mean}.
The reduction of variance is evident, in comparison to the uncontrolled case in Figure \ref{fig:free_plants}.

The ideal control scenario shown in Figure \ref{fig:control_prop_to_mean} is not realistic due to practical constraints in the field.
We address several of these challenges below:
\begin{itemize}
\item[i)] Daily application of nitrogen is not realistic due to costs and constraints in the availability of machinery. In Figure \ref{fig:control_lowsample} we explore the more realistic case of applications performed every 14 days, the ``sparse'' scenario. 
This better reflects the limitations of actuation and sensing.
\item[ii)] Computing the average growth $\bar{y}$ can also be expensive for a large field. This can be mitigated by using distributed control laws of the form
\begin{equation} \label{eqn:prop_local}
  \boldsymbol{u}_i = \bar{u} + \text{sat}\left(\dfrac{k}{|\mathcal{N}_i|}\sum_{j\in \mathcal{N}_i}(\boldsymbol{y}_j -\boldsymbol{y}_i)\right)
\end{equation}
where $\mathcal{N}_i$ is the set of plants that belong to the local neighbourhood of the plant $i$. 
In simulation, we assign the neighbourhood of each plant by placing the plants in a virtual grid.
The neighbours of plant $i$ are the adjacent and diagonal plants in the virtual grid, giving each plant up to $8$ neighbours.
Figure \ref{fig:control_lowsample_local} shows simulation results with both sparse applications and the local control law.
\item[iii)] Finally, noisy measurements must be taken into account in any realistic scenario. 
This case is explored in Figure \ref{fig:control_lowsample_local_and_noise}, where each observation $\mathbf{y}_i$ has additive Gaussian noise with standard deviation of $\sigma = 0.1 \mathbf{y}_i$.
\end{itemize}
These cases show a progressive degradation of performance. 
However, each modification results in a reduced variance in comparison to the uncontrolled case of Figure \ref{fig:free_plants}.

\begin{figure}
\begin{center}
  \subfloat[][Ideal control. \label{fig:control_prop_to_mean}] {\includegraphics[width=0.48\columnwidth]{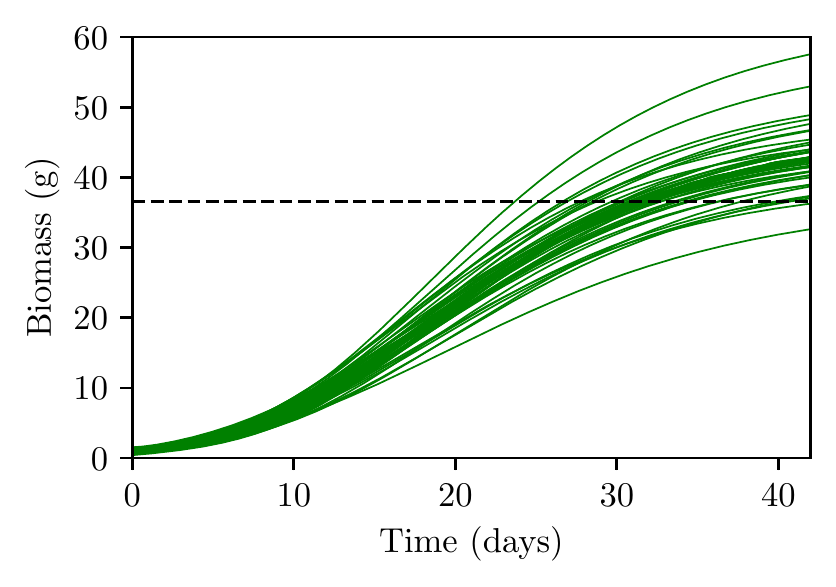}}
  \hfill 
  \subfloat[][Sparse control.\label{fig:control_lowsample}] {\includegraphics[width=0.48\columnwidth]{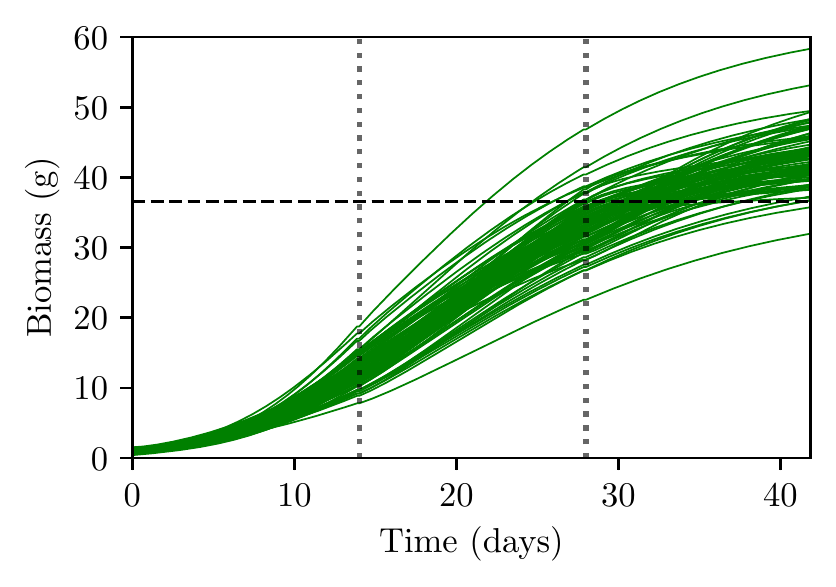}}
  \\
  \subfloat[][Sparse local control.
  \label{fig:control_lowsample_local}] {\includegraphics[width=0.48\columnwidth]{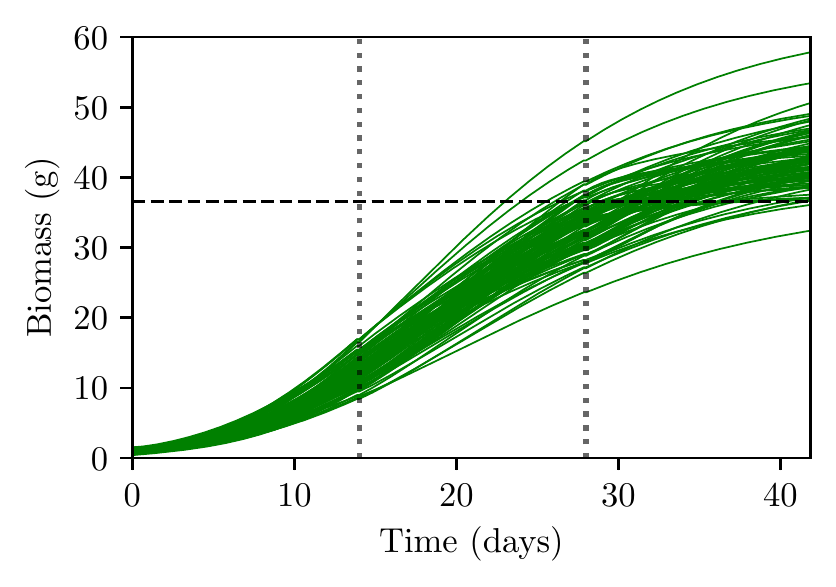}}     
  \hfill
  \subfloat[][Sparse local control with noise.
  \label{fig:control_lowsample_local_and_noise}] {\includegraphics[width=0.48\columnwidth]{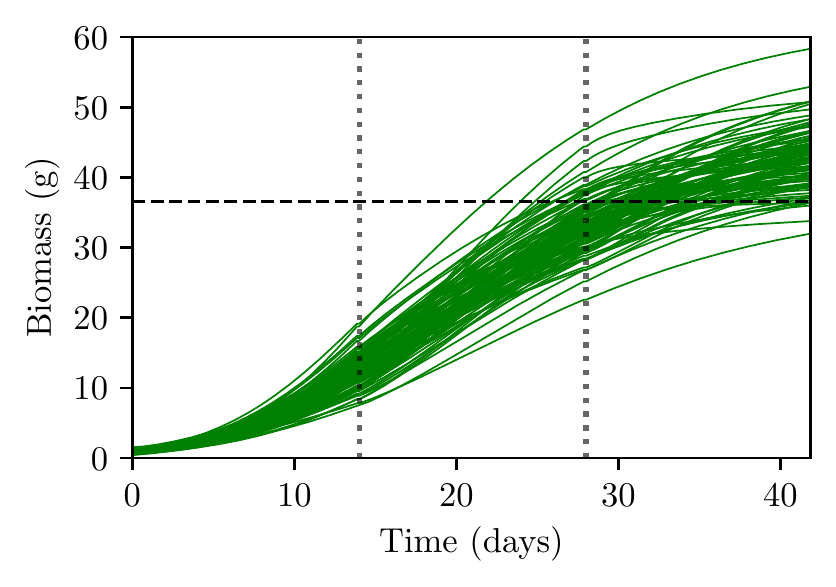}}     

  \caption{Control of the shoot biomass, $y=\psi\boldsymbol{b}$, in various simulation scenarios. In Figures \ref{fig:control_lowsample}, \ref{fig:control_lowsample_local}, and \ref{fig:control_lowsample_local_and_noise} the vertical dotted lines show when the sparse control updated the nitrogen available to the plant. The horizontal dashed line indicates the rejection threshold for the lettuce crop, set at $36.6$g the 10th percentile in the uncontrolled scenario.}
\end{center}
\end{figure}

The ideally controlled case has a significantly reduced variance, with the variance reduced from 
$29.73 \text{g}^2$ 
in the uncontrolled scenario to 
$7.91\text{g}^2$ 
in the ideally controlled scenario.  
As shown in figure \ref{fig:ideal_comparison}, this resulted in the portion of crop meeting the rejection threshold increasing from 90\% to 98\% without an increase in nitrogen availability.
In the simulation above, the ideal control scenario consumed 99\% of the nitrogen used in the uncontrolled scenario.

Figure \ref{fig:result_comparison} shows box plots summarising the final shoot biomass for the uncontrolled scenario, ideally controlled scenario, and the three sparsely controlled scenarios. 
The figure shows that the control performance was reduced in the sparse control scenarios with the use of the local mean controller, and the introduction of noise.
The sparse control scenario had a variance of 
$13.97\text{g}^2$, 
while the sparse local control had a variance of
$14.07\text{g}^2$. 
In both of these scenarios the portion of crop that met the threshold was 98\%, matching the ideal control scenario.
The scenario with added noise shows degraded performance in comparison to the ideal control case. 
However, as shown in Figure \ref{fig:noisy_comparison}, it still provides an improvement from the uncontrolled scenario, with a coefficient of variation of 
$17.76\text{g}^2$ 
and 94\% of plants meeting the rejection threshold.

\begin{figure}[h]
  \begin{center}
    \subfloat[][Uncontrolled vs Ideal control. \label{fig:ideal_comparison}] {\includegraphics[width=0.48\columnwidth]{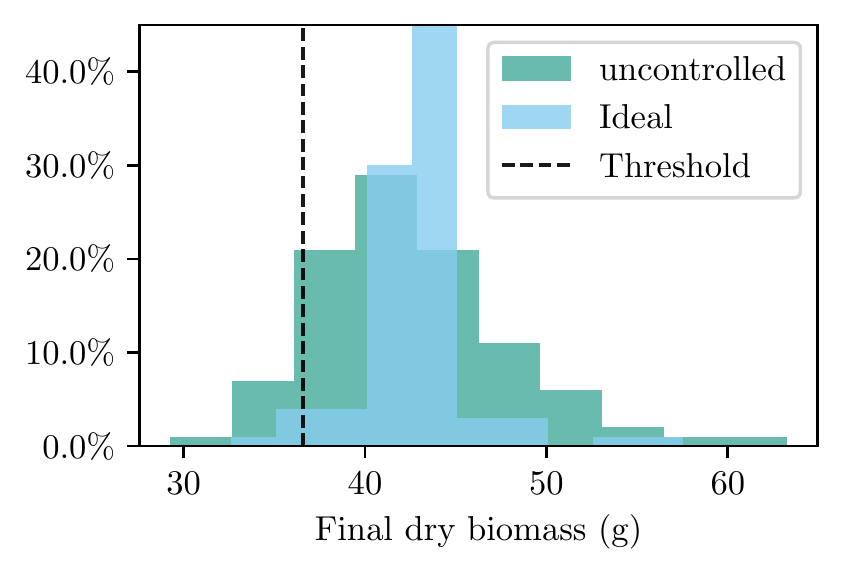}}
    \hfill 
    \subfloat[][Uncontrolled vs Sparse local control with noise.\label{fig:noisy_comparison}] {\includegraphics[width=0.48\columnwidth]{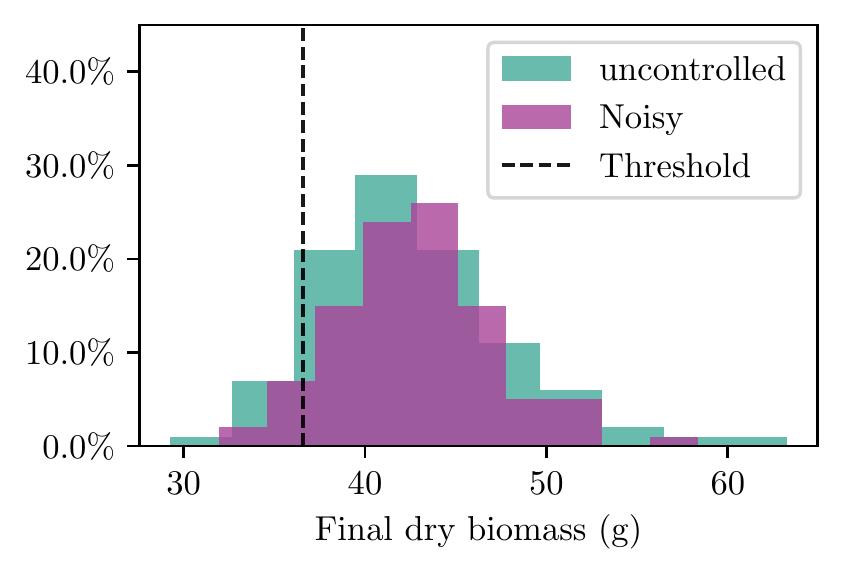}}
    \caption{Comparison between the histogram in the uncontrolled case, and the histograms for the ideally controlled and sparse local control with noise scenarios. }
  \end{center}
  \end{figure}

  \begin{figure}[h]
    \begin{center}
      \subfloat[][Uncontrolled vs Ideal vs Sparse cases with $\bar{u} = 0.075$g. \label{fig:result_comparison}] {\includegraphics[width=0.85\columnwidth]{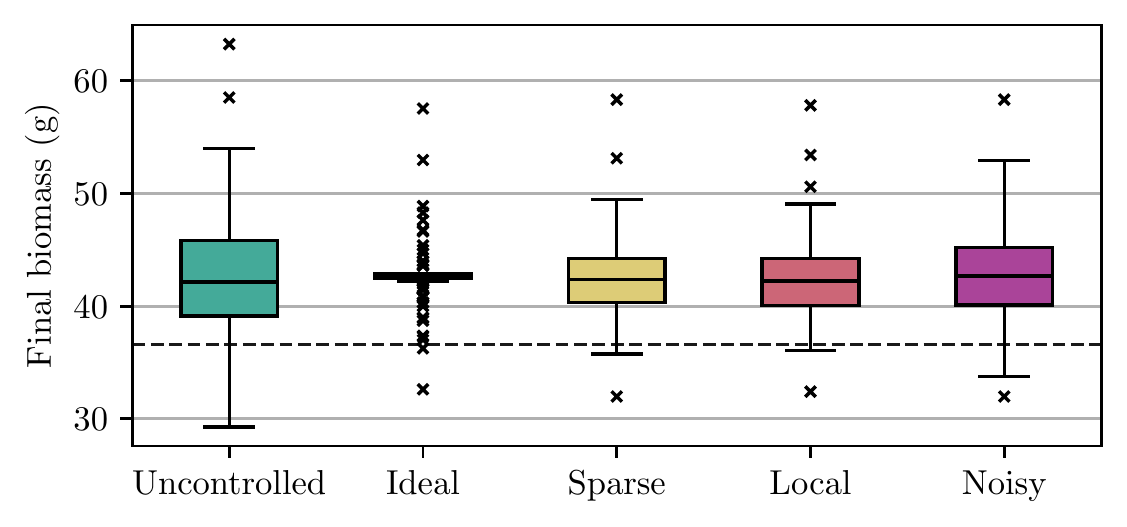}}\\
      \subfloat[][Uncontrolled vs Ideal vs Sparse cases with $\bar{u} = 0.073$g, $0.0655< u_i < 0.0805$. In the uncontrolled case $\bar{u} = 0.075$g is unchanged. \label{fig:boxplot_reduced}] {\includegraphics[width=0.85\columnwidth]{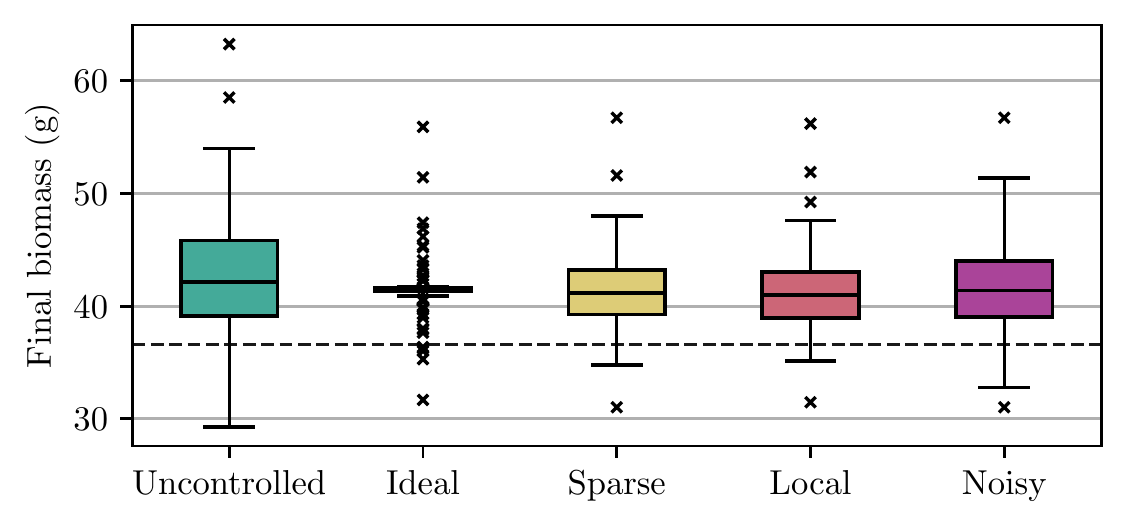}} 
      \caption{Boxplots for the Uncontrolled vs Ideal vs Sparse cases. In \ref{fig:boxplot_reduced} $\bar{u}$ is reduced demonstrating that an improvement in uniformity enables a reduction in the amount of nitrogen applied.}
  
    \end{center}
    \end{figure}

\addtolength{\textheight}{-3cm}   

The improvement in uniformity may be sufficient to reduce the average nitrogen input $\bar{u}$ to the field resulting in substantial reductions in nitrogen use without a loss in yield. 
This is explored in Figure \ref{fig:boxplot_reduced} where $\bar{u} = 0.073$g.
In the ideal and sparse control scenarios, 95\% and 94\% of lettuces respectively are above the rejection threshold, outperforming the uncontrolled scenario. 
In the local scenario 91\% of plants are above the threshold matching the uncontrolled scenario.
In the noisy scenario 90\% of plants were above the threshold, matching the uncontrolled scenario.
In all of these cases nitrogen use was substantially reduced, with the worst performing scenario consuming 4.8\% less nitrogen than in the uncontrolled case. 
For the farmer even a small percentage improvement in yield or decrease in fertiliser use is financially significant and an improvement in yield of 4\%, as in the noisy case with $\bar{u}=0.075$g, or the reduction of fertiliser use of $4.8$\% in this case is substantial.


\section{Conclusion} \label{section:conclusion}
In this paper we have presented a model for plant growth and preliminary simulation results controlling a field as a multi-agent system in simulation.
Our model is well suited to control of lettuce in an agriculture as it is a mechanistic model with a small state space and cooperative behaviour.
The model performs well when fitted to datasets from real plants.

Controlling a field as a multi-agent system and applying a basic distributed law to the individual plants is a promising method for the generation of individualised fertiliser prescriptions and decision making in precision agriculture.
When nitrogen was applied daily this control strategy improved the uniformity of the simulated crop significantly.
Uniformity was also improved in more realistic scenarios, where nitrogen was applied every two weeks and noise was introduced.

Our future work will focus on improving the controller design, providing performance guarantees for our system, and implementing the controller in a real agricultural process.
This will require the integration of supporting technologies such as remote sensing and variable rate application, and solving the related challenges of a real system with noisy, sparse data sources and infrequent actuation that have been discussed in simulation in this paper.





\bibliography{references.bib}

\end{document}